\documentclass[final,3p,times]{elsarticle}
\usepackage{lineno,hyperref}
\modulolinenumbers[5]

\usepackage{amsfonts,amssymb}
\usepackage{amsmath}
\usepackage{bm}
\usepackage{graphicx}
\usepackage{color}
\usepackage[ruled,vlined]{algorithm2e}
\usepackage{amsthm}
\usepackage{comment}

\usepackage{multicol}
\usepackage{float}
\usepackage{tabularx,booktabs}
\usepackage[table]{xcolor}
\usepackage{tikz}

\usepackage{mwe}
\usepackage{lettrine}
\usepackage{xcolor} 
\newtheorem{theorem}{Theorem}
\newtheorem{lemma}{Lemma}

\newtheorem{definition}{Definition}

\newtheorem{remark}{Remark}
\newtheorem{example}{Example}

\newtheorem{assumption}{Assumption}

\begin{document}

\begin{frontmatter}

\title{A graph-based spatial temporal logic for knowledge representation and automated reasoning 
\tnoteref{mytitlenote}}
\tnotetext[mytitlenote]{The partial support of the National Science Foundation (Grant No. CNS-1446288, ECCS-1253488, IIS-1724070) and of the Army Research Laboratory (Grant No. W911NF- 17-1-0072) is gratefully acknowledged.}



\author{Zhiyu Liu\corref{cor1}\fnref{label1}}
\ead{zliu9@nd.edu}
\author{Meng Jiang\fnref{label2}}
\author{Hai Lin\fnref{label1}}

\cortext[cor1]{Corresponding author}

\fntext[label1]{Department of Electrical Engineering, University of Notre Dame}
\fntext[label2]{Department of Computer Science and Engineering, University of Notre Dame}

\begin{abstract}
We propose a new graph-based spatial temporal logic for knowledge representation and automated reasoning in this paper. The proposed logic achieves a balance between expressiveness and tractability in applications such as cognitive robots. The satisfiability of the proposed logic is decidable. We apply a Hilbert style axiomatization for the proposed graph-based spatial temporal logic, in which Modus ponens and IRR are the inference rules. We show that the corresponding deduction system is sound and complete and can be implemented through SAT.
\end{abstract}

\begin{keyword}
Spatial temporal logic \sep Knowledge representation \sep Automated reasoning 
\end{keyword}

\end{frontmatter}


\section{Introduction}


Knowledge representation and automated reasoning are the foundations for achieving human-like artificial intelligence. On the one hand, knowledge representation builds an abstract model of the real world, which includes a set of multi-typed concepts and relations among them. On the other hand, automated reasoning searches for or even derives new information based on the abstract knowledge representation, which can be used to solve a wide range of real-world applications such as video interpretation \cite{smoliar1994content}, natural language process \cite{hirschberg2015advances}, and robot task planning \cite{fikes1971strips}.
Lots of work have been done on knowledge representation such as expert systems for solving specific tasks in the 1970s \cite{hayes1983building} and frame languages for rule-based reasoning in the 1980s \cite{kifer1995logical}. Researcher realized any intelligent process needs to be able to store knowledge in some forms and has the ability to reason on them with rules or logic. Currently, one of the most active models in knowledge representation is knowledge graph (or semantic web) such as Google knowledge graph  \cite{googleknowledgegraph} and ConceptNet \cite{liu2004conceptnet}. In the knowledge graph, concepts are modeled as nodes, and their relations are modeled as labeled edges. They make successful applications in areas such as recommendation systems and searching engine. However, the information in the knowledge graph could be inaccurate since contributors could be unreliable. They also face difficulties when describing the time and space sensitive information, which is particularly crucial to robots.

To avoid such difficulties, we, therefore, follow formal methods and logic-based approaches in our work.
In the formal method and logic-based approach, symbolic knowledge representation and reasoning are performed through primitive operations manipulating predefined elementary symbols \cite{hertzberg2008ai}. 
As one of the most investigated symbolic logics, first-order logics \cite{mccarthy1960programs} is a powerful representation and reasoning tool with a well-understood theory. It can be used to model various range of applications. However, first-order logic is, in general, undecidable. A sound and complete deduction algorithm cannot even guarantee termination, let alone real-time automatic reasoning \cite{hertzberg2008ai}. By limiting the expressiveness of first-order logic, some language subsets of first-order logic are decidable and have been used in applications, including software development and verification of software and hardware. 
One of the language subsets of first-order logic is propositional logics \cite{post1921introduction}. For many practical cases, the instances or the propositional variables are finite, which results in a decidable inference process. However, as one needs to consider all combinations of propositional variables, the growth is multiexponential in terms of the domain sizes of the propositional variables. Another language subset of first-order logic is description logic \cite{baader2003description}. In description logic, the domain representation is given by two parts, namely terminological knowledge and assertional knowledge. Terminological knowledge states properties of concepts and roles along with their relations. Assertional knowledge states concrete example of individuals. Description logic inferences are running based on the given terminological and assertional knowledge. The inference algorithms can run efficiently in most practical cases even though they are theoretically intractable. In general, classic logic fails to capture the temporal and spatial characteristics of the knowledge, and the inference algorithms are often undecidable. For example, it is difficult to capture information such as a robot hand is required to hold a cup for at least five minutes. 

The classic logics have sufficient expressive power on sequential planning. However, it may not be sufficient in terms of modeling temporal and spatial relations, such as the effects of action duration and space information from sensors. As spatial and temporal information are often particularly important for robots, spatial logic and temporal logic are studied both separately \cite{baier2008principles,cohn2001qualitative,raman2015reactive} and combined \cite{kontchakov2007spatial,haghighi2016robotic,bartocci2017monitoring}. Examples of temporal logics include Linear temporal logic \cite{baier2008principles} and Signal temporal logic \cite{raman2015reactive}. Examples of spatial logics include Region connection calculus (RCC) \cite{cohn1997qualitative} and $S4_u$ \cite{aiello2007handbook}. One can extend both temporal logic and spatial logic with metric extension such as interval algebra and rectangle algebra at the expense of computational complexity. Furthermore, there has been extensive work on combining temporal logic and spatial logic. Without any restriction on combining temporal predicates/operators and spatial predicates/operators (such as LTL and RCC), the obtained spatial temporal logic has the maximal expressiveness with an undecidable inference \cite{kontchakov2007spatial}. 
By integrating spatial and temporal operators with classic logic operators, spatial temporal logic shows excellent potential in specifying a wide range of task assignments with automated reasoning ability. Thus, in this paper, we are interested in employing spatial temporal logics for knowledge representation and developing automated reasoning based on it. One of the significant challenges comes from the complexity of the inference algorithm. Most spatial temporal logics enjoy immense expressiveness due to their semantic definition and the method of combining spatial and temporal operators \cite{kontchakov2007spatial}. However, their inference algorithms are often undecidable, and human inputs are often needed to facilitate the deduction process.  
The lack of tractability limits the application of spatial temporal logic. A balance between expressiveness and tractability is needed for spatial temporal logic.

Motivated by the challenges faced in existing work, we propose a new graph-based spatial temporal logic (GSTL) with a sound and complete inference system. GSTL is defined based on a hierarchical graph where both connectivity and parthood among spatial elements are considered. By introducing interval algebra in GSTL, it has abilities to specify directional relations between spatial elements and temporal relations between time intervals. The proposed GSTL can specify a wide range of spatial and temporal specifications while maintains a tractable inference algorithm because the interaction between temporal operators and spatial operators in the proposed GSTL is somewhat limited. 
The contributions of this paper are mainly focusing on proposing a new spatial temporal logic with a better balance between expressiveness and tractability. The satisfiability of the proposed GSTL is decidable. The inference systems are sound and complete and are implemented through SAT.
 


The rest of the paper is organized as follows. In Section \ref{Section: GSTL definitio}, we propose the graph-based spatial temporal logic, GSTL. 
We introduce the deduction system in Section \ref{Section: deduction system}, where the soundness and completeness are discussed. The implementation of automated reasoning is introduced in Section \ref{Section: implementation} with an example. 
Section \ref{Section: conclusion} concludes the paper.

\section{Graph-based spatial temporal logic}\label{Section: GSTL definitio}

This section aims to introduce the novel graph-based spatial temporal logic. We begin with a brief review of temporal representation and temporal logic.

\subsection{Temporal representation and logics}
\subsubsection{Temporal logics}
As the proposed spatial temporal logic is built by extending signal temporal logic with extra spatial operators, we first introduce the definition of signal temporal logic (STL) \cite{raman2015reactive}. 
\begin{definition}[STL Syntax]\rm
STL formulas are defined recursively as:
$$
\varphi::={\rm True}|\pi^\mu|\neg\pi^{\mu}|\varphi\land\psi|\varphi\lor\psi|\Box_{[a,b]} \psi | \varphi\sqcup_{[a,b]} \psi,
$$
where $\pi^\mu$ is an atomic predicate $\mathbb{R}^n\to\{0,1\}$ whose truth value is determined by the sign of a function $\mu:\mathbb{R}^n\to\mathbb{R}$, i.e., $\pi^\mu$ is true if and only if $\mu({\bf x})>0$; and $\psi$ is an STL formula. The ``eventually" operator $\Diamond$ can also be defined here by setting $\Diamond_{[a,b]} \varphi={\rm True}\sqcup_{[a,b]} \varphi$.
\end{definition}

The semantics of STL with respect to a discrete-time signal $\bf x$ are introduced as follows, where $({\bf x},t_k)\models \varphi$ denotes for which signal values and at what time index the formula $\varphi$ holds true.
\begin{definition}[STL Semantics]\rm
The validity of an STL formula $\varphi$ with respect to an infinite run ${\bf x}= x_0x_1x_2\ldots$ at time $t_k$ is defined inductively as follows.
\begin{enumerate}
\item $({\bf x},t_k)\models \mu$, if and only if $\mu(x_k)>0$;
\item $({\bf x},t_k)\models \neg\mu$, if and only if $\neg(({\bf x},t_k)\models \mu)$;
\item $({\bf x},t_k)\models \varphi\land\psi$, if and only if $({\bf x},t_k)\models \varphi$ and $({\bf x},t_k)\models \psi$;
\item $({\bf x},t_k)\models \varphi\lor\psi$, if and only if $({\bf x},t_k)\models \varphi$ or $({\bf x},t_k)\models \psi$;
\item $({\bf x},t_k)\models \Box_{[a,b]}\varphi$, if and only if $\forall t_{k'}\in[t_k+a,t_k+b]$, $({\bf x},t_{k'})\models \varphi$;
\item $({\bf x},t_k)\models \varphi\sqcup_{[a,b]}\psi$, if and only if $\exists t_{k'}\in[t_k+a,t_k+b]$ such that $({\bf x},t_{k'})\models \psi$ and $\forall t_{k''}\in[t_k,t_{k'}]$, $({\bf x},t_{k''})\models \varphi$;

\item $({\bf x},t_k)\models \Diamond_{[a,b]}\varphi$, if and only if $\exists t_{k'}\in[t_k+a,t_k+b]$, $({\bf x},t_{k'})\models \varphi$.
\end{enumerate}
\end{definition}

A run ${\bf x}$ satisfies $\varphi$, denoted by $\bf x\models\varphi$, if $({\bf x},t_0)\models\varphi$. Intuitively, $\bf x\models \Box_{[a,b]}\varphi$ if $\varphi$ holds at every time step between $a$ and $b$, ${\bf x}\models \varphi\sqcup_{[a,b]}\psi$ if $\varphi$ holds at every time step before $\psi$ holds and $\psi$ holds at some time step between $a$ and $b$, and ${\bf x}\models \Diamond_{[a,b]}\varphi$ if $\varphi$ holds at some time step between $a$ and $b$. An STL formula $\varphi$ is {\it bounded-time} if it contains no unbounded operators. The bound of $\varphi$ can be interpreted as the horizon of future predicted signals $\bf x$ that is needed to calculate the satisfaction of $\varphi$.

\subsubsection{Temporal representation}
The flow of time we are interested in is a non-empty set of time points with a partial irreflexive ordering. We formally define a flow of time with a set of time points and a partial ordering as follows.
\begin{definition}[Flow of time]
We define the flow of time as a non-empty set for time points with a connected partial ordering $(T,<)$ where $T$ is time and $<$ is the irreflexive partial order. $(T,<)$ is said to be connected if and only if for all $x,y\in T$ with $x<y$ there is a finite sequence such that $x=x_0<x_1<...<x_n=y, ~\forall i\in[1,...,n],~ x_i\in T$.
\end{definition}

There are multiple ways to represent time, e.g., continuous time, discrete time, and interval. In this paper, we use a discrete time interval. In order to increase the expressiveness of the proposed spatial temporal logic, we also consider Allen interval algebra to extend the until operator in STL.

\begin{definition}[Allen interval algebra\cite{allen1984towards}]
Allen interval algebra defines the following 13 temporal relationships between two intervals, namely before ($b$), meet ($m$), overlap ($o$), start ($s$), finish ($f$), during ($d$), equal ($e$), and their inverse ($^{-1}$) except equal. The 13 temporal relationships are illustrated in Fig. \ref{Allen interval algebra}. 
\begin{figure}
    \centering
    \includegraphics[scale=0.28]{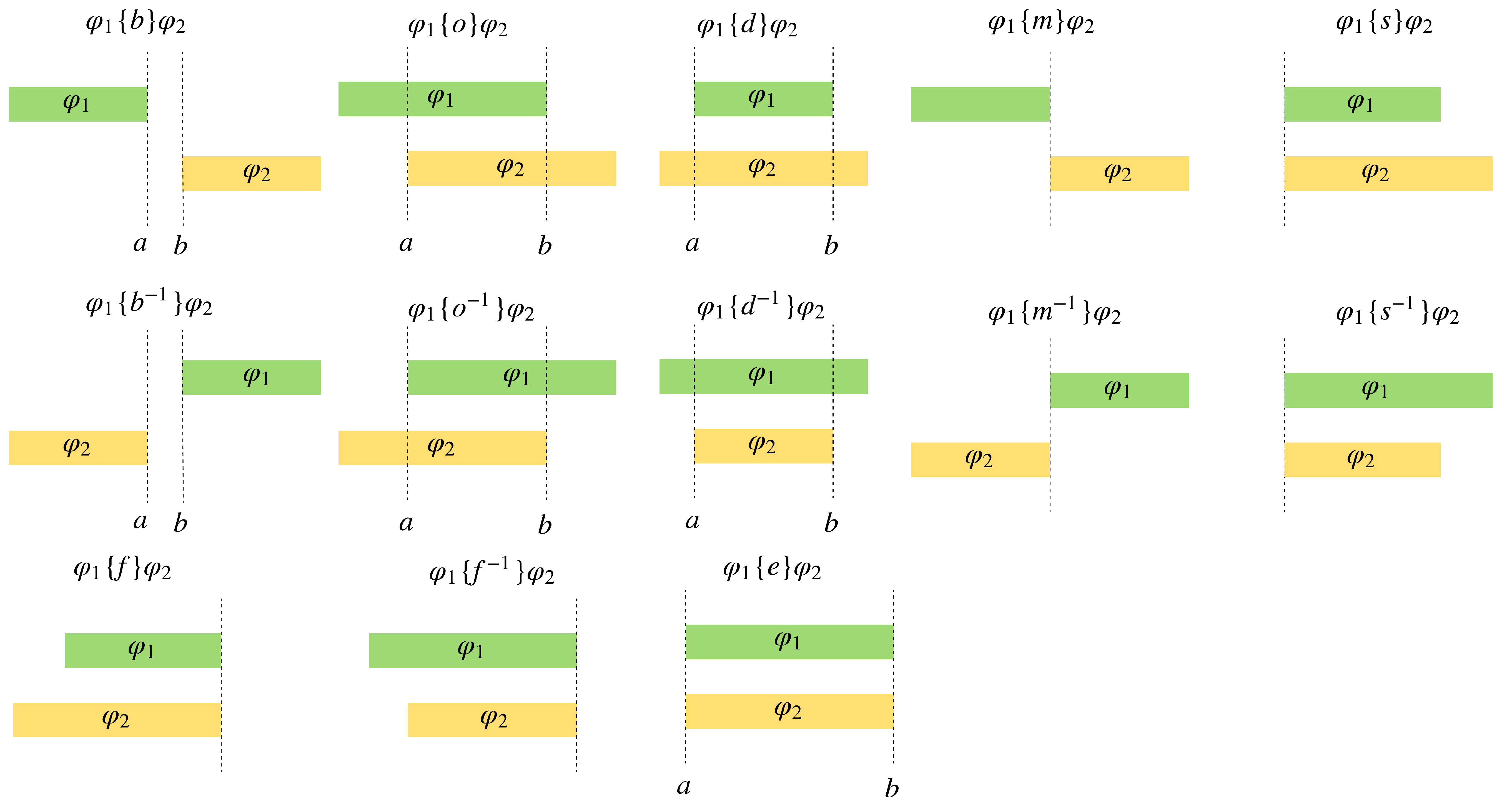}
    \caption{The 13 basic qualitative relations in interval algebra}
    \label{Allen interval algebra}
\end{figure}
\end{definition}

\subsection{Spatial model}

After reviewing temporal representation and logic, we next take a look at the spatial model we are interested in. The basic spatial characteristics we considered in our spatial model are, namely, ontology, mereotopology, and metric spatial representation. 
\subsubsection{Spatial ontology}
We use regions as the basic spatial elements instead of points. Within the qualitative spatial representation community, there is a strong tendency to take regions of space as the primitive spatial entity \cite{cohn2001qualitative}. In practice, a reasonable constraint to impose would be that regions are all rational polygons.

\subsubsection{Mereotopology}
As for the relations between regions, we consider mereotopology, meaning we consider both mereology (parthood) and topology (connectivity) in our spatial model. Parthood describes the relational quality of being a part. For example, \emph{wheel is a part of car} and \emph{cup is one of the objects in a cabinet}. Connectivity describes if two spatial objects are connected. For example, \emph{hand grabs a cup}. By considering mereotopology, the proposed GSTL will have more expressive power than existing spatial temporal logic STREL \cite{bartocci2017monitoring} and SpaTeL \cite{haghighi2015spatel}.

We apply a graph with a hierarchy structure to represent the spatial model.
Denote $\Omega=\cup_{i=1}^{n}\Omega_i$ as the union of the sets of all possible spatial objects where $\Omega_i$ represents a certain set of spatial objects or concepts.
\begin{definition}[Graph-based Spatial Model]
The graph-based spatial model with a hierarchy structure $\mathcal{G}=(\mathcal{V},\mathcal{E})$  is constructed by the following rules.
\begin{itemize}
    \item The node set $\mathcal{V}=\{V_1,...,V_n\}$ is consisted of a group of node set where each node set $V_k$ represents a finite subset spatial objects from $\Omega_i$. At each layer, $V_k=\{v_{k,1},...,v_{k,n_k}\}$ contains nodes which represent $n_k$ spatial objects in $\Omega_i$.
    
    \item The edge set $\mathcal{E}$ is used to model the relationship between nodes such as whether two nodes are adjacent or if one node is included within another node. $e_{i,j}\in \mathcal{E}$ if and only if $v_i$ and $v_j$ are connected via two relations defined below.
    
    \item $v_{k,i}$ is a \emph{parent} of $v_{k+1,j}$, $\forall k\in[1,...,n-1]$, if and only if $v_{k,i} \wedge v_{k+1,j}= v_{k+1,j}$. $v_{k+1,j}$ is called a \emph{child} of $v_{k,i}$ if $v_{k,i}$ is its parent. $v_0$ is the only node that does not have a parent. All nodes in $V_n$ do not have child. Furthermore, if $v_i$ and $v_j$ are a pair of parent-child, then $e_{i,j}\in\mathcal{E}$. $v_i$ is a \emph{neighbor} of $v_j$ and $e_{i,j}\in\mathcal{E}$ if and only if there exist $k$ such that $v_i\in V_k$, $v_j\in V_k$, and the minimal distance between $v_i$ and $v_j$ is less than a given threshold $\epsilon$.  
\end{itemize}
\end{definition}

\begin{example}\rm 
An example is given in Fig. \ref{exp:Graph with a hierarchy structure} to illustrate the proposed spatial model. In Fig. \ref{exp:Graph with a hierarchy structure}, $V_1=\{kitchen\}$, $V_2=\{body~part,~ tool,~ material\}$ and $V_3=\{head,~ hand,~ cup,~ bowl,~ table,~ milk,~ butter\}$. The parent-child relationships are drawn in solid lines, and the neighbor relationships are drawn in dashed lines. Each layer represents the space with different spatial concepts or objects by taking categorical values from $\Omega_i$, and connections are built between layers. The hierarchical graph is able to express facts such as ``head is part of a body part," and ``cup holds milk."

\end{example}

\begin{figure}
 \centering
 \includegraphics[scale=0.7]{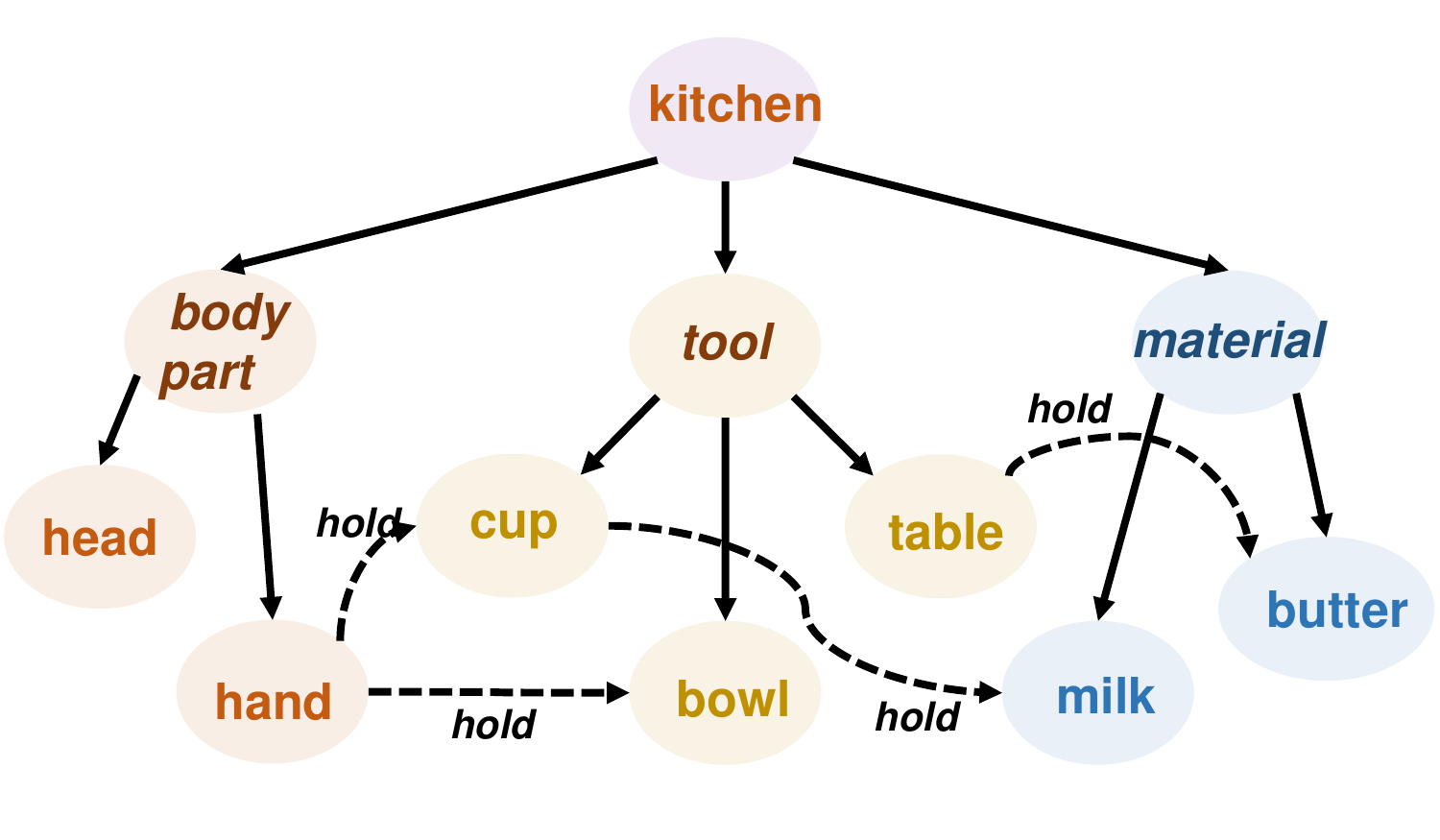}
 \caption{The hierarchical graph with three basic spatial operators: parent, child, and neighbor where the parent-child relations are drawn in solid line and the neighbor relations are drawn in dash line.}
 \label{exp:Graph with a hierarchy structure}
\end{figure}

\subsubsection{Metric Spatial representation}
To further increase the expressiveness of the proposed GSTL for cognitive robots, we include directional information in our spatial model. It is done by extending rectangle algebra into 3D, which is more suitable for cognitive robots. We first briefly introduce rectangle algebra below. 

\begin{definition}[Rectangle algebra \cite{smith1992algebraic}]
In rectangle algebra, spatial objects are considered as rectangles whose sides are parallel to the axes of some orthogonal basis in a 2D Euclidean space. The $13\times 13$ basic relations between two spatial objects are defined by extending interval algebra in 2D.
\begin{align*}
    \mathcal{R}_{RA}=\{(A,B):A,B\in \mathcal{R}_{IA}\},
\end{align*}
where $\mathcal{R}_{IA}$ is the set containing 13 interval algebra relations. 

\end{definition}
Rectangle algebra extends interval algebra into 2D. It can be used to expressive directional information such as left, right, up, down, and their combination. However, rectangle algebra is defined in 2D only while cognitive robots are often deployed in a 3D environment. Thus, in this paper, we extend it to 3D.
\begin{align*}
    \mathcal{R}_{CA}=\{(A,B,C):A,B,C\in \mathcal{R}_{IA}\},
\end{align*}
where $13\times 13\times 13$ basic relations are defined for cubic algebra (CA). An example is given in Fig. \ref{Relations in CA} to illustrate the cubic algebra. For spatial objects $X$ and $Y$ in the left where $X$ is at the front, left, and below of $Y$, we have $X\{(b,b,o)\}Y$. For spatial objects $X$ and $Y$ in the right where $Y$ is completely on top of $X$, we have $X\{(e,e,m)\}Y$.
\begin{figure}
    \centering
    \includegraphics[scale=0.3]{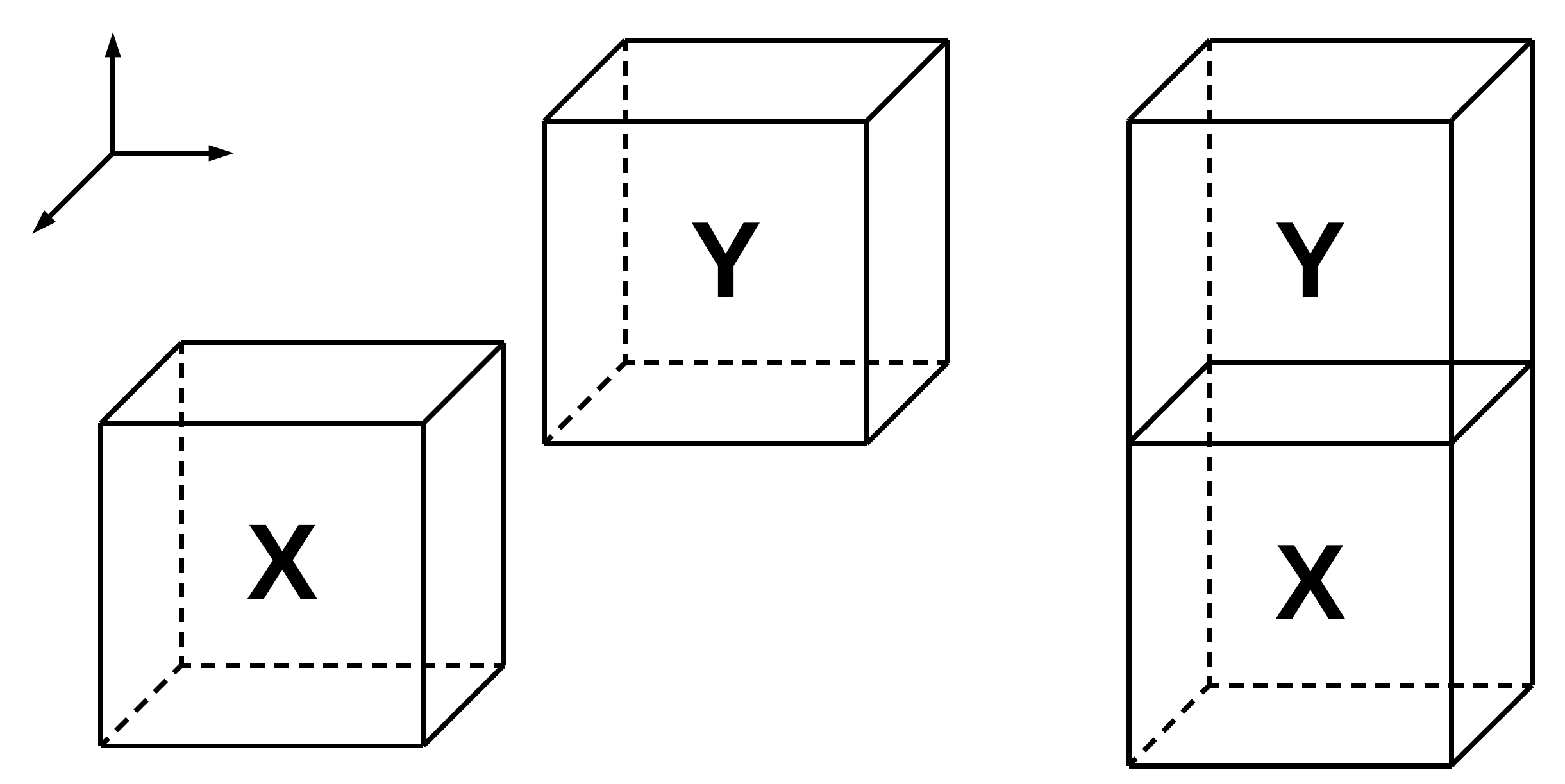}
    \caption{Representing directional relations between objects $X$ and $Y$ in CA}
    \label{Relations in CA}
\end{figure}

\subsubsection{Spatial temporal signals}
The spatial temporal signals we are interested in are defined as follow.
\begin{definition}[Spatial Temporal Signal]
A spatial temporal signal $x(v,t)$ is defined as a scalar function for node $v$ at time $t$
\begin{equation}
    x(v,t): V\times T \rightarrow D,
\end{equation}
where $D$ is the signal domain. Depends on different applications, $D$ can be a Boolean domain, real-value domain, and categorical domain and etc. 
\end{definition}

\subsection{Graph-based spatial temporal logic}
With the temporal model and spatial model in mind, we now give the formal syntax and semantics definition of GSTL in this section. GSTL is defined based on a hierarchy graph by combining STL and three spatial operators where the until operator and the neighbor operator are enriched by interval algebra (IA) and cubic algebra (CA), respectively. 

\begin{definition}[GSTL Syntax]\label{definition: GSTL syntax}
The syntax of a GSTL formula is defined recursively as 
\begin{equation}
    \begin{aligned}
    &\tau := \mu ~|~ \neg\tau ~|~ \tau_1\wedge\tau_2 ~|~ \tau_1\vee\tau_2 ~|~ \mathbf{P}_A\tau ~|~ \mathbf{C}_A \tau ~|~ \mathbf{N}_A^{\left \langle *, *, * \right \rangle} \tau,\\
    &\varphi := \tau ~|~ \neg \varphi ~|~ \varphi_1\wedge\varphi_2 ~|~ \varphi_1\vee\varphi_2 ~|~ \Box_{[a,b]} \varphi~| ~ \varphi_1\sqcup_{[a,b]}^{*}\varphi_2,
    \end{aligned}
    \label{complexity proof 1}
\end{equation}
where $\tau$ is a spatial term and $\varphi$ is the GSTL formula. $\mu$ is an atomic predicate (AP), negation $\neg$, conjunction $\wedge$ and disjunction $\vee$ are the standard Boolean operators. Spatial operators are ``parent" $\mathbf{P}_A$, ``child" $\mathbf{C}_A$, and ``neighbor" $\mathbf{N}_A^{\left \langle *, *, * \right \rangle}$, where $A$ denotes the set of nodes which they operate on and $*\in\{b, o, d, \equiv, m, s, f\}$ is one of the seven temporal relationships defined in the Allen interval algebra. $\Box_{[a,b]}$ is the ``always" operator and $\sqcup_{[a,b]}^{*}$ is the ``until" temporal operators with an Allen interval algebra extension, where $[a,b]$ being a real positive closed interval and $*\in\{b, o, d, \equiv, m, s, f\}$.
\end{definition}

\begin{remark}
We only consider a subclass of CA and IA relations, namely convex IA relations. Convex CA relations are composed exclusively of convex IA relations, which is defined in \cite{ligozat1996new}. For example, $\{b, m, o\}$ is a convex IA relation while $\{b, o\}$ is not. It has been shown that the spatial reasoning on convex relations can be solved in polynomial time using constraint programming while the reasoning with the full CA/IA expressiveness is NP-complete.
\end{remark}

Intuitively, the parent operator $\mathbf{P}_A$ describes the behavior of the parent of the current node. The child operator $\mathbf{C}_A$ describes the behavior of children of the current node in the set $A$. The neighbor operator $\mathbf{N}_A^{\left \langle *, *, * \right \rangle}$ describes the behavior of neighbors of the current node in the set $A$. 

Before we give the semantics definition of GSTL, we first define an interpretation function. The interpretation function $\iota(\mu,x(v,t)): AP\times D \rightarrow R$ interprets the spatial temporal signal as a number based on the given atomic proposition $\mu$. The qualitative semantics of the GSTL formula is given as follows.

\begin{definition}[GSTL Qualitative Semantics]
The satisfiability of a GSTL formula $\varphi$ with respect to a spatial temporal signal $x(v,t,\iota)$ at time $t$ and node $v$ is defined inductively as follows.
\begin{enumerate}
\item $x(v,t,\iota)\models \mu$, if and only if $\iota(\mu,x(v,t))>0$;
\item $x(v,t,\iota)\models \neg\varphi$, if and only if $\neg(x(v,t,\iota))\models \varphi)$;
\item $x(v,t,\iota)\models \varphi\land\psi$, if and only if $x(v,t,\iota)\models \varphi$ and $x(v,t,\iota)\models \psi$;
\item $x(v,t,\iota)\models \varphi\lor\psi$, if and only if $x(v,t,\iota)\models \varphi$ or $x(v,t,\iota)\models \psi$;
\item $x(v,t,\iota)\models \Box_{[a,b]}\varphi$, if and only if $\forall t'\in[t+a,t+b]$, $x(v,t',\iota)\models \varphi$;
\item $x(v,t,\iota)\models \Diamond_{[a,b]}\varphi$, if and only if $\exists t'\in[t+a,t+b]$, $x(v,t',\iota)\models \varphi$;
\end{enumerate}
The until operator with interval algebra extension is defined as follow.
\begin{enumerate}
\item $x(v,t,\iota)\models \varphi\sqcup_{[a,b]}^b\psi$, if and only if $x(v,t,\iota)\models\Box_{[a,b]}\neg(\varphi\vee\psi)$ and $\exists t_1<a,~\exists t_2>b$ such that $x(v,t,\iota)\models\Box_{[t_1,a]}(\varphi\wedge\neg\psi)\wedge\Box_{[b,t_2]}(\neg\varphi\wedge\psi)$;
\item $x(v,t,\iota)\models \varphi\sqcup_{[a,b]}^o\psi$, if and only if $x(v,t,\iota)\models\Box_{[a,b]}(\varphi\wedge\psi)$ and $\exists t_1<a,~\exists t_2>b$ such that $x(v,t,\iota)\models\Box_{[t_1,a]}(\varphi\wedge\neg\psi)\wedge\Box_{[b,t_2]}(\neg\varphi\wedge\psi)$;
\item $x(v,t,\iota)\models \varphi\sqcup_{[a,b]}^d\psi$, if and only if $x(v,t,\iota)\models\Box_{[a,b]}(\varphi\wedge\psi)$ and $\exists t_1<a,~\exists t_2>b$ such that $x(v,t,\iota)\models\Box_{[t_1,a]}(\neg\varphi\wedge\psi)\wedge\Box_{[b,t_2]}(\neg\varphi\wedge\psi)$;
\item $x(v,t,\iota)\models \varphi\sqcup_{[a,b]}^\equiv\psi$, if and only if $x(v,t,\iota)\models\Box_{[a,b]}(\varphi\wedge\psi)$ and $\exists t_1<a,~\exists t_2>b$ such that $x(v,t,\iota)\models\Box_{[t_1,a]}(\neg\varphi\wedge\neg\psi)\wedge\Box_{[b,t_2]}(\neg\varphi\wedge\neg\psi)$;
\item $x(v,t,\iota)\models \varphi\sqcup^m\psi$, if and only if $\exists t_1<t<t_2$ such that $x(v,t,\iota)\models\Box_{[t_1,t]}(\varphi\wedge\neg\psi)\wedge\Box_{[t,t_2]}(\neg\varphi\wedge\psi)$;
\item $x(v,t,\iota)\models \varphi\sqcup^s\psi$, if and only if $\exists t_1<t<t_2$ such that $x(v,t,\iota)\models\Box_{[t_1,t]}(\neg\varphi\wedge\neg\psi)\wedge\Box_{[t,t_2]}(\varphi\wedge\psi)$;
\item $x(v,t,\iota)\models \varphi\sqcup^f\psi$, if and only if $\exists t_1<t<t_2$ such that $x(v,t,\iota)\models\Box_{[t_1,t]}(\varphi\wedge\psi)\wedge\Box_{[t,t_2]}(\neg\varphi\wedge\neg\psi)$;
\end{enumerate}
The spatial operators are defined as follows.
\begin{enumerate}
\item $x(v,t,\iota)\models \mathbf{P}_A\tau$, if and only if $\forall v_p\in A,~x(v_p,t,\iota)\models \tau$ where $v_p$ is the parent of $v$;
\item $x(v,t,\iota)\models \mathbf{C}_{A}\tau$, if and only if $\forall v_c\in A,~x(v_c,t,\iota)\models \tau$ where $v_c$ is a child of $v$;
\item $x(v,t,\iota)\models \mathbf{N}_{A}^{\left \langle b, *, * \right \rangle}\tau$, if and only if $\forall v_n \in A,~x(v_n,t,\iota)\models \tau$ where $v_n$ is a neighbor of $v$ and $v_n[x^+]<v[x^-]$;
\item $x(v,t,\iota)\models \mathbf{N}_{A}^{\left \langle o, *, * \right \rangle}\tau$, if and only if $\forall v_n \in A,~x(v_n,t,\iota)\models \tau$ where $v_n$ is a neighbor of $v$ and $v_n[x^-]<v[x^-]<v_n[x^+]<v[x^+]$;
\item $x(v,t,\iota)\models \mathbf{N}_{A}^{\left \langle d, *, * \right \rangle}\tau$, if and only if $\forall v_n \in A,~x(v_n,t,\iota)\models \tau$ where $v_n$ is a neighbor of $v$ and $v_n[x^-]<v[x^-]<v[x^+]<v_n[x^+]$;
\item $x(v,t,\iota)\models \mathbf{N}_{A}^{\left \langle \equiv, *, * \right \rangle}\tau$, if and only if $\forall v_n \in A,~x(v_n,t,\iota)\models \tau$ where $v_n$ is a neighbor of $v$ and $v_n[x^-]=v[x^-],~v[x^+]=v_n[x^+]$;
\item $x(v,t,\iota)\models \mathbf{N}_{A}^{\left \langle m, *, * \right \rangle}\tau$, if and only if $\forall v_n \in A,~x(v_n,t,\iota)\models \tau$ where $v_n$ is a neighbor of $v$ and $v_n[x^+]=v[x^-]$;
\item $x(v,t,\iota)\models \mathbf{N}_{A}^{\left \langle s, *, * \right \rangle}\tau$, if and only if $\forall v_n \in A,~x(v_n,t,\iota)\models \tau$ where $v_n$ is a neighbor of $v$ and $v_n[x^-]=v[x^-]$;
\item $x(v,t,\iota)\models \mathbf{N}_{A}^{\left \langle f, *, * \right \rangle}\tau$, if and only if $\forall v_n \in A,~x(v_n,t,\iota)\models \tau$ where $v_n$ is a neighbor of $v$ and $v_n[x^+]=v[x^+]$,
\end{enumerate}
where $v[x^-]$ and $v[x^+]$ denote the lower and upper limit of node $v$ in x-direction. 
\end{definition}
Definition for the neighbor operator in y-direction and z-direction is omitted for simplicity. Notice that the reverse relations in IA and CA can be easily defined by changing the order of the two GSTL formulas involved, e.g., $\varphi\sqcup_{[a,b]}^{o^{-1}}\psi\Leftrightarrow\psi\sqcup_{[a,b]}^o\varphi$. As usual, $\varphi_1\rightarrow\varphi_2$, $\varphi_1\leftrightarrow\varphi_2$ abbreviate $\neg\varphi_1\vee\varphi_2$, $(\varphi_1\rightarrow\varphi_2)\wedge(\varphi_2\rightarrow\varphi_1)$, respectively.
We further define another six spatial operators $\mathbf{P}_{\exists}\tau$, $\mathbf{P}_{\forall}\tau$, $\mathbf{C}_{\exists}\tau$, $\mathbf{C}_{\forall}\tau$, $\mathbf{N}_{\exists}^{\left \langle *, *, * \right \rangle}\tau$ and $\mathbf{N}_{\forall}^{\left \langle *, *, * \right \rangle}\tau$ based on the definition above. 
\begin{align*}
    &\mathbf{P}_{\exists}\tau=\wedge_{i=1}^{n_p}\mathbf{P}_{A_i}\tau,
    ~\mathbf{P}_{\forall}\tau=\vee_{i=1}^{n_p}\mathbf{P}_{A_i}\tau,~A_i=\{v_{p,i}\},\\
    &\mathbf{C}_{\exists}\tau=\wedge_{i=1}^{n_c}\mathbf{C}_{A_i}\tau,
    ~\mathbf{C}_{\forall}\tau=\vee_{i=1}^{n_c}\mathbf{C}_{A_i}\tau,~A_i=\{v_{c,i}\},\\
    &\mathbf{N}_{\exists}^{\left \langle *, *, * \right \rangle}\tau=\wedge_{i=1}^{n_n}\mathbf{N}_{A_i}^{\left \langle *, *, * \right \rangle}\tau,
    ~\mathbf{N}_{\forall}^{\left \langle *, *, * \right \rangle}\tau=\vee_{i=1}^{n_n}\mathbf{N}_{A_i}^{\left \langle *, *, * \right \rangle}\tau,~A_i=\{v_{n,i}\},
\end{align*}
where $v_{p,i}$, $v_{c,i}$, $v_{n,i}$ are the parent, child, and neighbor of $v$ respectively and $n_p$, $n_c$, and $n_n$ are the number of parents, children and neighbors of $v$ respectively.

The proposed GSTL is able to express various temporal and spatial properties. For example, the until operator with Allen interval algebra extension can express temporal relations between two intervals, as shown in Fig. \ref{Enrich until operator with IA}. In Fig. \ref{Relations in CA}, we can use the defined neighbor operator to represent the spatial relations for $X$ and $Y$ as $X\models\mathbf{N}_\exists^{\left \langle b, b, o \right \rangle}Y$ and $X\models\mathbf{N}_\exists^{\left \langle e, e, m \right \rangle}Y$.
\begin{figure}
    \centering
    \includegraphics[scale=0.3]{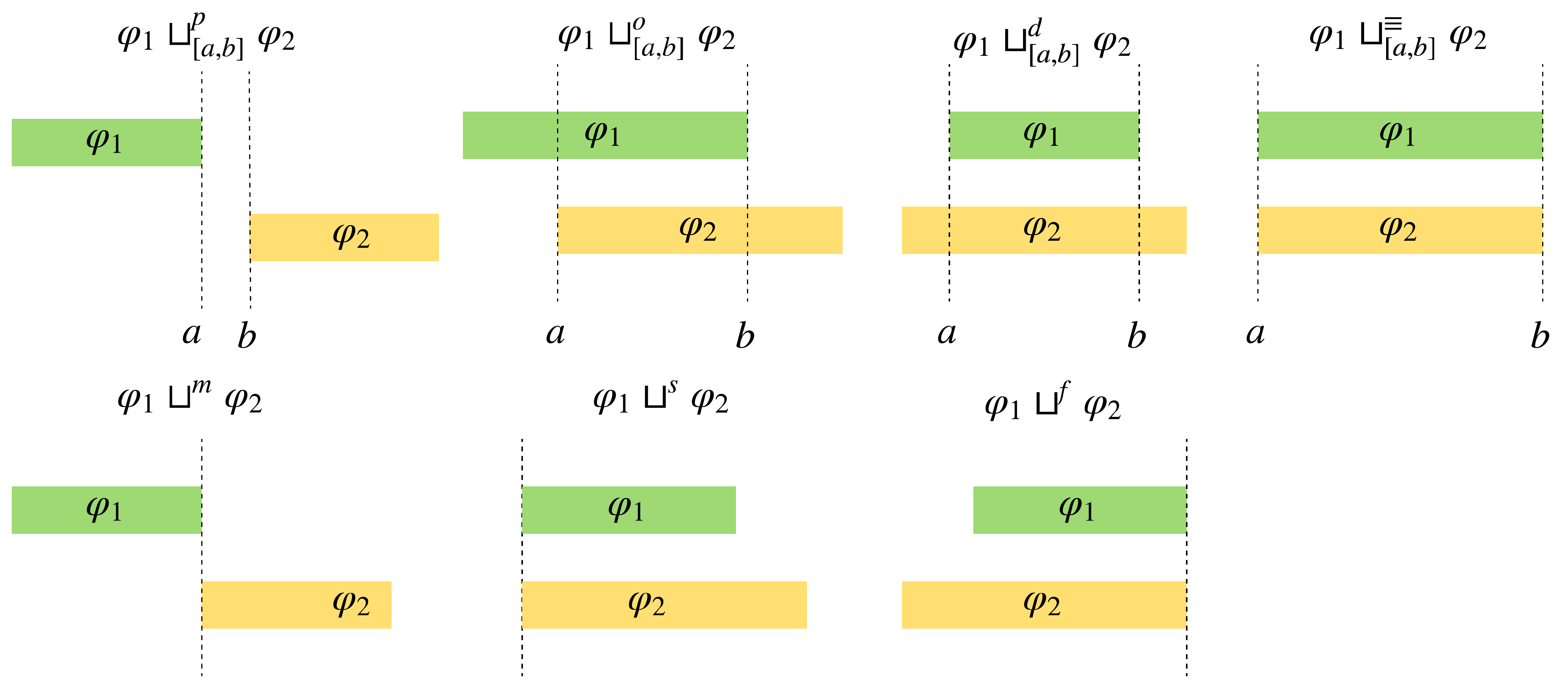}
    \caption{Enrich until operator with IA}
    \label{Enrich until operator with IA}
\end{figure}

\subsection{Expressiveness and tractability of GSTL}

We investigate the expressiveness and tractability of the proposed spatial temporal logic and compare it with respect to existing several classic temporal and spatial logic, namely  $S4_u$, $RCC-8$, STL, and $S4_u\times LTL$.
First, we compare the expressiveness with existing spatial logics $S4_u$ and $RCC-8$. $RCC-8$ \cite{smith1992algebraic} was introduced in geographical information systems as a decidable subset of Region Connection Calculus ($RCC$). $RCC-8$ studies region variables and defines eight binary relations among the variables. If we assume regions are rectangles in 2D or cubic in 3D, all eight relations can be expressed by the neighbor operator in GSTL.  For example, the equal relation $EQ(a,b)$ in $RCC-8$ can be represented as $a\models\mathbf{N}_\exists^{\left \langle e, e, e \right \rangle}b$. The parent/child relations (e.g. $fork\models\mathbf{P}_\exists tools$) and directional information (e.g. left and right) in GSTL cannot be expressed by $RCC-8$. $S4_u$ is a well known propositional modal logic that has strictly larger expressive power than $RCC-8$ \cite{kontchakov2007spatial}. All four atoms (subset, negation, conjunction, and disjunction) in $S4_u$ formulas can be expressed by GSTL. However, the interior and closure operators in spatial terms of $S4_u$ cannot be expressed in GSTL. Similar to $RCC-8$, directional information in GSTL cannot be expressed by $S4_u$. 

Then, we compare GSTL with a popular temporal logic STL \cite{raman2015reactive}. As we mentioned before, GSTL is defined by extending STL through enriching the until operator with interval algebra. Compared to STL, GSTL is able to express more information between two temporal intervals, such as overlap and during which cannot be expressed by STL. 

In the end, we compare GSTL with existing spatial temporal logics, including $S4_u\times LTL$, SpaTeL, and STREL. $S4_u\times LTL$ are defined by combing two modal logics where no restrictions are added on their spatial and temporal predicates. Due to the freedom of combining spatial and temporal operators, $S4_u\times LTL$ enjoys a powerful expressiveness where the spatial terms can change over time. It can express statement such as ``an egg will eventually turn into a chicken". However, even though $S4_u$ and LTL are decidable, the satisfiability problem for $S4_u\times LTL$ is not decidable. Compared to $S4_u\times LTL$, GSTL has less expressive power in the sense that the spatial term cannot change over time. However, the complexity of the satisfiability problem for GSTL is decidable, which is discussed below. 

One of the fundamental problems for any logic is the satisfiability of a finite set of formulas. Specifically,  the satisfiability problem given a set of GSTL formulas is defined as the problem of deciding if the formulas are satisfiable or consistent by finding a spatial temporal signal which can satisfy all spatial temporal constraints specified by the GSTL formulas. The satisfiability problem is important as the deduction problem discussed in the following section can often be transferred to the satisfiability problem. The most important algorithmic properties of the satisfiability problem is its computational complexity. We show that the complexity of the satisfiability problem for GSTL is decidable.
First, we adopt the following assumption, which is reasonable for applications such as cognitive robots.
\begin{assumption}[Domain closure]\label{assumption: domain clousre}
The only objects in the domain are those representable using the existing symbols, which do not change over time.
\end{assumption}
Then we give the computational complexity of the satisfiability problem for the proposed GSTL in the following theorem.
\begin{theorem}
The computational complexity of the satisfiability problem for the proposed GSTL formulas is NP-complete, and the finite sets of satisfiable constraints are recursively enumerable.
\label{theorem:complexity}
\end{theorem}
\begin{proof}
From the Eq. \eqref{complexity proof 1} in Definition \ref{definition: GSTL syntax}, we can see that the interactions between spatial operators and temporal operators are rather limited. Temporal operators are not allowed in any spatial terms. Thus, for every GSTL formula $\varphi$ we can construct an new formula $\varphi^*$ by replacing every occurrence of a spatial sub-formula $\tau$ ($\mathbf{P}_A\tau, ~ \mathbf{C}_A \tau, ~ \mathbf{N}_A^{\left \langle *, *, * \right \rangle} \tau$) in $\varphi$ as shown in \eqref{complexity proof 1} with a new propositional variable $\mu_\tau$. Then we obtain a formula without spatial operator as shown below, and it is a bounded STL formula with interval algebra extension.
\begin{equation}
    \varphi = \mu_\tau ~|~ \neg \varphi ~|~ \varphi_1\wedge\varphi_2 ~|~ \varphi_1\vee\varphi_2 ~|~ \Box_{[a,b]} \varphi~| ~ \varphi_1\sqcup_{[a,b]}^{*}\varphi_2.
    \label{GSTL*}
\end{equation}
Now the problem transfer to the complexity of the satisfiability problem for a bounded STL formula with interval algebra extension. We formulate the problem as a Boolean satisfiability problem (SAT) by using SAT encoding recursively according to the definition of GSTL \cite{liu2017distributed}. We show that any GSTL formulas can be reformed in conjunctive normal form (CNF). Here, formulas in CNF are a conjunction of one or more clauses and each clause is a disjunction of literals which are AP or their negation. Based on the definition of CNF, it is straightforward to see that GSTL formulas $\varphi=\mu_r$, $\varphi=\neg\phi$, $\varphi=\varphi_1\wedge\varphi_2$, and $\varphi=\varphi_1\vee\varphi_2$ are already in CNF. For GSTL formula with ``Always" operator $\varphi=\Box_{[a,b]}\phi$, we have $\varphi=\wedge_{i=a}^b\phi_i$ where $\phi_i$ represents $\phi$ at time $i\in[a,b]$. For GSTL formula with the until operator extended with the interval algebra, we give the encoding procedure for $\sqcup_{[a,b]}^o$. The rest can be encoded using the same procedure. $\varphi_1\sqcup_{[a,b]}^{o}\varphi_2$ can be encoded as $\wedge_{i=a}^{b}(\varphi_{1,i}\wedge\varphi_{2,i})\wedge \varphi_{1,a-1} \wedge \varphi_{2,b+1} \wedge \neg\varphi_{1,b+1} \wedge \neg\varphi_{2,a-1}$. From the above encoding procedure, we can see that the SAT encoding resulting a Boolean satisfiability problem. It has been shown that the complexity of a SAT problem is NP-complete according to Cook–Levin theorem \cite{cook1971complexity} which is generally considered as decidable. Existing heuristic SAT-algorithms are able to solve formulas consisting of millions of symbols, which is sufficient for many practical SAT problems.
\end{proof}

\begin{remark}
The restriction that no temporal operators are allowed in the spatial term is reasonable for applications in robotics since normally predicates are used to represent objects such as cups and bowls. We don't expect cups to change to bowls over time. Thus, we don't need any temporal operator in the spatial term and adopt Assumption \ref{assumption: domain clousre}.
\end{remark}

\section{Automated reasoning based on GSTL}\label{Section: deduction system}

In the previous section, we introduced the syntax and semantics definition of GSTL and showed that it is suitable for knowledge representation since the satisfiability of the proposed GSTL is decidable, and it demonstrates reasonably expressive power.  In this section, we focus on automatic reasoning. In particular, we adopt Hilbert style axiomatization for the proposed GSTL, where the proof system is composed of a set of axioms and several inference rules. We choose Hilbert style axiomatization due to its simple set of inference rules and no distinction between formulas and conclusions \cite{gabbay1990axiomatization}. The axioms are generated through a predefined set of axiom schemas, which are defined as below.

\begin{definition}[Axiom schemas]
Axiom schemas are axiom templates that represent infinitely many specific instances of axioms by replacing the variables with any syntax valid formulas. The variables ranging over formulas are called schematic variables.
\end{definition}
For example, for a set of atomic propositions in $\{a,b,c,...\}$, we can get the axiom $(\Box_{[0,\infty)}\neg a \wedge b)\rightarrow (\neg a \wedge b)$ from the axiom schema $\Box_{[0,\infty)}\varphi\rightarrow\varphi$. We denote a set of axiom schemas as $Z$. The axioms we can get depends on both $Z$ and the set of atomic propositions.

\subsection{Axiomatization system}
We define the axiomatization system with a set of axiom schemas and inference rules given below. There are three parts in the axiom schemas, namely propositional logics (P), temporal logics (T), and spatial logics (S). P1 to P10 are axiom schemas from propositional logic. 
\begin{equation}
    \begin{aligned}
    &P1~\neg\neg\varphi \Rightarrow \varphi,\\
    &P2~\varphi_1,\varphi_2 \Leftrightarrow \varphi_1\wedge\varphi_2,\\
    &P3~\varphi_i \Rightarrow \varphi_1\vee\varphi_2,\\
    &P4~\varphi_1\rightarrow(\varphi_2\rightarrow\varphi_1),\\
    &P5~(\phi \rightarrow(\psi \rightarrow \xi)) \rightarrow((\phi \rightarrow \psi) \rightarrow(\phi \rightarrow \xi)),\\
    &P6~(\neg \phi \rightarrow \neg \psi) \rightarrow(\psi \rightarrow \phi),\\
    &P7~\varphi_1\vee\varphi_2, \varphi_1\rightarrow\varphi_3, \varphi_2\rightarrow\varphi_3 \Rightarrow \varphi_3,\\
    &P8~\varphi_1\rightarrow\varphi_2, \varphi_2\rightarrow\varphi_1 \Leftrightarrow \varphi_1\leftrightarrow\varphi_2,\\
    &P9~\neg(\varphi_1\wedge\varphi_2)\Leftrightarrow (\neg\varphi_1)\vee(\neg\varphi_2),\\
    &P10~\neg(\varphi_1\vee\varphi_2)\Leftrightarrow (\neg\varphi_1)\wedge(\neg\varphi_2).
    \end{aligned}
    \label{axiom schemas1}
\end{equation}
T1 to T5 are axiom schemas for temporal logics. 
\begin{equation}
    \begin{aligned}
    &T1~\Box_{[a,b]}(\varphi_1\rightarrow\varphi_2)\Rightarrow \Box_{[a,b]}\varphi_1\rightarrow\Box_{[a,b]}\varphi_2,\\
    &T2~\Box_{[a,b]}(\varphi_1\wedge\varphi_2)\Leftrightarrow\Box_{[a,b]}\varphi_1 \wedge \Box_{[a,b]}\varphi_2,\\
    &T3~\Diamond_{[a,b]}(\varphi_1\wedge\varphi_2)\Rightarrow\Diamond_{[a,b]}\varphi_1 \wedge \Diamond_{[a,b]}\varphi_2\Rightarrow \Diamond_{[a,b]}(\varphi_1\vee\varphi_2),\\
    &T4~(\varphi_1\wedge\varphi_2)\sqcup_{[a,b]}^*\varphi_3 \Leftrightarrow \varphi_1\sqcup_{[a,b]}^*\varphi_3 \wedge \varphi_2\sqcup_{[a,b]}^*\varphi_3,\\
    &T5~\varphi_1\sqcup_{[a,b]}^*(\varphi_2\wedge\varphi_3) \Leftrightarrow \varphi_1\sqcup_{[a,b]}^*\varphi_2 \wedge \varphi_1\sqcup_{[a,b]}^*\varphi_3.
    \end{aligned}
    \label{axiom schemas2}
\end{equation}
S1 to S6 are axiom schemas for spatial logics.
\begin{equation}
    \begin{aligned}
    &S1~\mathbf{P}_A (\varphi_1\wedge\varphi_2) \Leftrightarrow \mathbf{P}_A \varphi_1 \wedge \mathbf{P}_A \varphi_2,\\
    &S2~\mathbf{P}_A (\varphi_1\vee\varphi_2) \Leftrightarrow \mathbf{P}_A \varphi_1 \vee \mathbf{P}_A \varphi_2,\\
    &S3~\mathbf{C}_A (\varphi_1\wedge\varphi_2) \Leftrightarrow \mathbf{C}_A \varphi_1 \wedge \mathbf{C}_A \varphi_2,\\
    &S4~\mathbf{C}_A (\varphi_1\vee\varphi_2) \Leftrightarrow \mathbf{C}_A \varphi_1 \vee \mathbf{C}_A \varphi_2,\\
    &S5~\mathbf{N}_A^{\left \langle *, *, * \right \rangle} (\varphi_1\wedge\varphi_2) \Leftrightarrow \mathbf{N}_A^{\left \langle *, *, * \right \rangle} \varphi_1 \wedge \mathbf{N}_A^{\left \langle *, *, * \right \rangle} \varphi_2,\\
    &S6~\mathbf{N}_A^{\left \langle *, *, * \right \rangle} (\varphi_1\vee\varphi_2) \Leftrightarrow \mathbf{N}_A^{\left \langle *, *, * \right \rangle} \varphi_1 \vee \mathbf{N}_A^{\left \langle *, *, * \right \rangle} \varphi_2.
    \end{aligned}
    \label{axiom schemas3}
\end{equation}
The inference rules are:
\begin{equation}
    \begin{aligned}
    &\textbf{Modus ponens}~ \frac{\varphi_1,\varphi_1\rightarrow\varphi_2}{\varphi_2},\\
    &\textbf{IRR(irreflexivity)}~
    \frac{\mu\vee \Diamond_{[t,b]} \mu\rightarrow\varphi}{\varphi}, \text{for all formulas $\varphi$ and atoms $\mu$ not appearing in $\varphi$ at current time $t$}.
    \end{aligned}
    \label{inference rules}
\end{equation}

\subsection{Properties of the axiomatization system}
Next, we show that the proposed axiomatization system containing axiomatization defined above is sound and complete. 
First, we define a proof for GSTL inferences.
\begin{definition}[Proof]
A proof in GSTL is a finite sequence of GSTL formulas $\varphi_1,\varphi_2,...,\varphi_n$, where each of them is an axiom or there exists $i,j<k$, such that $\varphi_k$ is the conclusion derived from $\varphi_i$ and $\varphi_j$ using the inference rules in the axiomatization system. The GSTL formula $\varphi_n$ is the conclusion of the proof, and $n$ is the length of the proof.
\end{definition}

Now, we discuss the soundness of the proposed axiomatization system.
Informally soundness means if in all formulas that are possible given a set of formulas $\Sigma$, the formula $\varphi$ also holds. Formally we say the axiomatization system is sound if a set $\Sigma$ of well-formed formulas semantically implies a certain well-formed formula $\varphi$ if all truth assignments that satisfy all the formulas in $\Sigma$ also satisfy $\varphi$.
\begin{theorem}[Soundness]\label{theorem:soundness}
The above axiomatization system is sound for the given spatial temporal signals $x(v,t,\iota)$.
\end{theorem}
\begin{proof}
This can be proved from the fact that all axioms' schemas from propositional logic, temporal logic, and spatial logic are valid, and all rules preserve validity. To prove soundness, we aim to prove the following statement. Given a set of GSTL formulas $\Sigma$, any formula $\phi$, which can be inferred from the axiomatization system above, is correct, meaning all truth assignments that satisfy $\Sigma$ also satisfy $\phi$. 

The proof is done by induction. First, if $\phi\in \Sigma$, then it is trivial to say $\phi$ is correct. Second, if $\phi$ belongs to one of the axiom schemas, it is also trivial to say $\phi$ is correct since all axiom schemas defined in \eqref{axiom schemas1}-\eqref{axiom schemas3} preserve semantic implication based on the semantic definition of GSTL. We need to further prove the inference rules are also sound.

Let us assume if $\phi_i$ can be proved by $\Sigma$ in $n$ steps in a proof, then $\phi_i$ is implied by $\Sigma$. For each possible application of a rule of inference at step $i + 1$, leading to a new formula $\phi_j$, if we can show that $\phi_j$ can be implied by $\Sigma$, we prove the soundness. If Modus ponens is applied, then $\phi_i=\varphi_1\wedge(\varphi_i\rightarrow\varphi_2)$ and $\phi_j=\varphi_2$. Let $x(v,t,\iota)$ be a spatial temporal signal with $t\in(T,<)$ being a flow of time with a connected partial ordering and $\iota$ being the interpretation function. Assuming $x(v,t,\iota)\models\phi_i$, we need to prove that $x(v,t,\iota)\models\phi_j$. According to the semantic definition of GSTL, $x(v,t,\iota)\models\phi_i\Leftrightarrow x(v,t,\iota)\models\varphi_1 ~\text{and}~ x(v,t,\iota)\models\varphi_1\rightarrow\varphi_2$. Since $\varphi_1\rightarrow\varphi_2\Leftrightarrow\neg\varphi_1\vee\varphi_2$, $x(v,t,\iota)\models\varphi_2$ which proves Modus pones is sound. If IRR is applied, then $\phi_i=(\mu\vee\Diamond_{[t,b]}\mu)\rightarrow\varphi$ and $\phi_j=\varphi$. Assuming $x(v,t,\iota)\models\phi_i$, we need to prove that $x(v,t,\iota)\models\phi_j$. According to the semantic definition of GSTL, the following equation hold true.
\begin{equation}
\begin{aligned}
    x(v,t,\iota)&\models(\mu\vee\Diamond_{[t,b]}\mu)\rightarrow\varphi \\
    &\models \neg(\mu\vee\Diamond_{[t,b]}\mu)\vee\varphi\\
    &\models (\neg\mu \wedge \Box_{[t,b]}\neg\mu) \vee \varphi.
\end{aligned}
\end{equation}
Since $x(v,t,\iota)\models\Box_{[t,b]}\mu \Leftrightarrow \forall t'\in [t,b], x(v,t',\iota)\models\mu$, thus $x(v,t,\iota)\not\models\neg\mu \wedge \Box_{[t,b]}\neg\mu$. Thus, $x(v,t,\iota)\models\varphi$ has to be true. That proves that the inference rule IRR is also sound. This completes the proof.

\end{proof}

Before we discuss completeness, we give several necessary definitions on terms that will be used later.

\begin{definition}
Definition of signature, theory, prove, $Z$-consistent, complete $L$-theory, and IRR model are given below.
\begin{enumerate}
    \item Signature: A signature $L$ is a countable set containing all atomic propositions.
    \item Theory: A $L$-theory is a set of GSTL formulas where all atomic propositions are from signature $L$.
    \item Prove: A theory $\Delta$ prove $\varphi$, denoted as $\Delta\vdash\varphi$, if for some finite $\Delta_0\subseteq\Delta$, $\vdash(\wedge\Delta_0)\rightarrow\varphi$.
    \item  $Z$-consistent theory: Denote $Z$ as a set of schemas. A theory $\Delta$ is said to be $Z$-consistent, or just consistent if $Z$ is understood, if and only if $\Delta\not\vdash\bot$.
    \item Complete $L$-theory: Let $L$ be a signature containing all atomic propositions. A theory $\Delta$ is said to be a complete $L$-theory if for all formulas $\varphi$, either $\varphi\in\Delta$ or $\neg\varphi\in\Delta$. $\Delta$ is said to be an IRR theory if for some atom $\mu$, $\mu\wedge\Box_{[a,b]}\neg\mu\in\Delta$ and if $\Diamond_{[a,b]}(\varphi_1\wedge(\varphi_2\wedge ... \Diamond_{[a,b]}\varphi_m)...)\in\Delta$, then $\Diamond_{[a,b]}(\varphi_1\wedge(\varphi_2\wedge ... \Diamond_{[a,b]}\varphi_m\wedge(\mu\vee\Box_{[t,b]}\mu))...)\in\Delta$.
    \item IRR model: A model $N=(S,\sqsubset)$ for any Z-consistent set of formulas is an IRR model if every instance of the Z-schema is valid in $N$ and $\mu\vee\Diamond_{[t,b]}$ is valid in $N$ for all $t$.
\end{enumerate}
\end{definition}

We first prove the following lemma stating the fact that for any $Z$-consistent theory, there is a set of complete $L$-theory containing it.
\begin{lemma}
Let $\Sigma$ be a Z-consistent theory with $L$ being its signature. Let $L^*\supseteq  L$ be an extension of $L$ by a countable infinite set of atoms. Then there is a Z-consistent complete IRR $L^*$-theory $\Delta$ containing $\Sigma$.
\label{lemma1}
\end{lemma}

\begin{proof}
We prove the lemma by inductively building a series of Z-consistent $L^*$-theories $\Delta_i, i<\omega$. Then their union $\Delta=\cup_i\Delta_i$ is the Z-consistent complete IRR $L^*$-theory $\Delta$.

Let $\mu$ be an atom of $L^*\setminus L$. Then $\mu$ does not appear in $\Sigma$. According to the proof in Theorem \ref{theorem:soundness}, $\Sigma\cup\{\mu\vee\Diamond_{[t,b]}\mu\}$ is Z-consistent. Set $\Delta_0=\Sigma\cup\{\mu\vee\Diamond_{[t,b]}\mu\}$. We denote all $L^*$-formulas as $\varphi_0,\varphi_1,...$. We assume $\Delta_i$ is constructed inductively and is a Z-consistent theory. We expend $\Delta_i$ through the following steps.
\begin{enumerate}
    \item If $\Delta_i\cup\{\varphi_i\}$ is not a Z-consistent theory, then $\Delta_{i+1}=\Delta_i\cup\{\neg\varphi_i\}$. $\Delta_{i+1}$ is a Z-consistent theory since the definition of the $Z$-consistent theory and  $\Delta_i$ is a Z-consistent theory based on our assumption.
    \item If $\Delta_i\cup\{\varphi_i\}$ is a Z-consistent theory and $\varphi_i$ is not of the form $\Diamond_{[a,b]}(\phi_1\wedge(\phi_2\wedge ... \Diamond_{[a,b]}\phi_m)...)$, then $\Delta_{i+1}=\Delta_i\cup\{\varphi_i\}$.
    \item If $\Delta_i\cup\{\varphi_i\}$ is a Z-consistent theory and $\varphi_i$ is of the form $\Diamond_{[a,b]}(\phi_1\wedge(\phi_2\wedge ... \Diamond_{[a,b]}\phi_m)...)$, then $\Delta_{i+1}=\Delta_i\cup\{\varphi_i,\Diamond_{[a,b]}(\phi_1\wedge(\phi_2\wedge ... \Diamond_{[a,b]}\phi_m\wedge(\mu\vee\Diamond_{[t,b]}\mu))...)$ which is a Z-consistent theory.
\end{enumerate}
Based on the definition of complete $L$-theory and IRR theory, $\Delta=\cup_{i<\omega}\Delta_i$ is a Z-consistent complete IRR $L^*$-theory.
\end{proof}

Another lemma is needed to prove the completeness, which states the fact that there exists an IRR model for a complete $L$-theory.
\begin{lemma}
For all $Z$-schema instances $\varphi\in\Delta$, we have
\begin{equation}
   \varphi ~\text{is valid in}~ N=(S,\sqsubset), \Delta\in S \Leftrightarrow \varphi\in\Delta.
\end{equation}
where $\Delta$ is a complete Z-consistent IRR L-theory. $N$ is an L-structure $N=(S,\sqsubset)$, where $\sqsubset$ is a binary relation \cite{gabbay1990axiomatization} on the set of all complete Z-consistent IRR L-theory defined as
\begin{equation}
    \Delta_1\sqsubset\Delta_2, \text{if and only if for all L-formulas, if~} \Box_{[a,b]}\varphi\in\Delta_1, \text{then~} \varphi\in\Delta_2.
\end{equation}
\label{lemma2}
\end{lemma}

\begin{proof}
It is proved by the induction of building $\varphi$. If $\varphi=\mu$, then based on the semantic definition of GSTL, the lemma is satisfied. If $\varphi=\neg\mu$ and $\varphi=\mu_1\wedge\mu_2$, it is obvious the lemma is satisfied based on the definition of completeness of $\Delta$.

Now assume $\varphi$ satisfies $\varphi$ is valid in $N$ if and only if $\varphi\in\Delta$. We first prove the lemma holds for $\Box_{[a,b]}\varphi$. 1) If $\Box_{[a,b]}\varphi$ is valid in $N$ for $\Delta_1\in S$, then there is $\Delta_2\in S$ with $\Delta_1\sqsubset\Delta_2$ and $\varphi$ is valid in $N$ for $\Delta_2$. Since $\varphi$ satisfies $\varphi$ is valid in $N$ if and only if $\varphi\in\Delta$, we have $\varphi\in\Delta_2$. If $\neg\Box_{[a,b]}\varphi\in\Delta_1$, then based on the definition of $\sqsubset$, $\neg\varphi\in\Delta_2$ which contradict our assumption. Thus, $\Box_{[a,b]}\varphi\in\Delta_1$. 2) Now assume that $\Box_{[a,b]}\varphi\in\Delta_1$. Then there is $\Delta_2\in S$ with $\Delta_1\sqsubset\Delta_2$ and $\varphi\in\Delta_2$ where $S$ is the set of all complete Z-consistent IRR L-theories. Since $\Delta_2\in S$ and $\varphi$ is valid in $N$ for $\Delta_2$ according to our inductive hypothesis,  $\Box_{[a,b]}\varphi$ is valid in $N$ for $\Delta_1$ based on the semantic definition of $\Box_{[a,b]}\varphi$.

As for the until operator with IA extension, we prove $\varphi_1\sqcup_{[a,b]}^o\varphi_2$ as an example and leave the rest to readers. They can all be proved using the same procedure. Based on the semantic definition of $\sqcup_{[a,b]}^o$, $\varphi_1\sqcup_{[a,b]}^o\varphi_2=\Box_{[a,b]}(\varphi_1\wedge\varphi_2)\wedge \Box_{[t_1,a]}(\varphi_1\wedge\neq\varphi_2)\wedge \Box_{[b,t_2]}(\neg\varphi_1\wedge\varphi_2)$ for some $t_1<a<b<t_2$. As we already show that the lemma holds true for $\varphi_1\wedge\varphi_2$ and $\Box_{[a,b]}\varphi$ inductively, the lemma holds true for $\Box_{[a,b]}(\varphi_1\wedge\varphi_2)\wedge \Box_{[t_1,a]}(\varphi_1\wedge\neg\varphi_2)\wedge \Box_{[b,t_2]}(\neg\varphi_1\wedge\varphi_2)$ as well.
\end{proof}

Now we can have the following theorem stating that if $\Sigma$ is a Z-consistent theory, then $\Sigma$ has an IRR modal where every instance of Z-schema is valid.

\begin{theorem}
Let $L$ be a countable infinite signature and $\Sigma$ be a Z-consistent L-theory. Then there are countable $L\subseteq L^*$ and an IRR model $N$ of $\Sigma$ such that every $L^*$-axiom of $Z$ is valid in $N$.
\label{theorem1}
\end{theorem}

\begin{proof}
The proof is based on Lemma \ref{lemma1} and Lemma \ref{lemma2}. Let us assume $\Sigma$ to be Z-consistent. Let $L^*$ to be augmented $L$ with countable infinite new atoms. Let $S$ be the set of all Z-consistent complete IRR $L^*$ theories. According to Lemma \ref{lemma1}, there is $\Delta\in S$ containing $\Sigma$. We define an IRR model $N=(S,\sqsubset)$ the same as the one in Lemma \ref{lemma2}, where $S$ is the set of all Z-consistent complete IRR $L^*$ theories and $\sqsubset$ is the binary relations on the theories in $S$. Then based on Lemma \ref{lemma2}, $N$ is an IRR model, and for all instances of Z-schemas are in every $\Delta\in S$, they are all valid in $N$. 
\end{proof}

An inference system for logic is complete with respect to its semantics if it produces a proof for each provable statement. If the semantics of a set of GSTL formulas $\Sigma$ implies $\varphi$, then $\varphi$ is proved by $\Sigma$. We formally defined completeness as follow.
\begin{definition}
The axiomatization system composed with axiom schemas $Z$ and inference rules is said to be complete if for all theories $\Sigma$, $\Sigma$ is Z-consistent if and only if there is an IRR model $N$ with the flow of time in $(T,<)$ such that $N$ is a model of $\Sigma$.
\end{definition}

The axiomatization system satisfies the following property.
\begin{theorem}[Completeness]
The axiomatization system for GSTL is complete with an IRR model $N$ such that $N$ is a model of $\Sigma$.
\end{theorem}

\begin{proof}
According to the soundness Theorem \ref{theorem:soundness}, if $\Sigma$ has a model, then it is consistent. Then according to Theorem \ref{theorem1}, if $\Sigma$ is consistent, then it has an IRR model. Thus, the axiomatization system is complete.
\end{proof}

\begin{remark}
The compactness theorem that any consistent theory has a model if each of its finite subsets does is a result that holds for first-order logic, temporal logic with the flow of time $\mathbb{Q}$. But it fails in the temporal logic for $\mathbb{R}$, $\mathbb{N}$, and $\mathbb{Z}$ \cite{gabbay1990axiomatization}. Since strong completeness implies compactness, we only discuss weak completeness theorem in the cases of $\mathbb{N}$ and $\mathbb{Z}$ where given appropriate schemas, any consistent formula has a model with the appropriate flow of time  \cite{gabbay1990axiomatization}. The strong completeness is equal to frame completeness and compactness in universal modal logic. However, temporal logic in the flow of real-time has weak completeness, which has finitely complete and expressively complete but does not have compactness. 
\end{remark}

\section{Implementation of the automated reasoning} \label{Section: implementation}

There are different automated reasoning tasks for an intelligent system, such as deriving new formulas from a given set of formulas, checking the consistency of a set of formulas, updating a given set of formulas based on new information, and finding a minimal representation. Although these tasks seem different, they can be reformed into each other \cite{cohn2008qualitative}. As algorithms developed for one problem can be easily implemented to solve other problems, most of the research in automated reasoning focuses on the consistency checking problem. We demonstrate how do we solve a consistency checking problem using SAT based on the proposed GSTL formulas and the axiomatization system defined above.

Assume we have $n$ GSTL formulas $\Sigma=\{\phi_1,\phi_2,...,\phi_n\}$ whose truthfulness are known. We aim to check if $\Sigma$ is consistent (e.g., $\Sigma\not\vdash\bot$). This is done by assigning every spatial term a truth value. It has been shown in the proof of Theorem \ref{theorem:complexity}, any GSTL formulas can be reformulated as CNF. Following the SAT encoding procedure in the proof of Theorem \ref{theorem:complexity}, we first eliminate all temporal operators and then reformulate the GSTL formulas in $\Sigma$ in the CNF form $\wedge_{j}^{p_1}(\vee_{i}^{p_2}\mu_{j,i}^*)$ where $\mu_{j,i}^*$ is the spatial term and $p_1$ and $p_2$ are temporal parameters in GSTL formulas. It has been shown that every propositional formula can be converted into an equivalent formula that is in CNF. It is done by applying the axiomatization system in \eqref{axiom schemas1}-\eqref{inference rules}. Then we obtain a set of logic constraints for spatial terms $\mu_{j,i}^*$ whose truth value are to be assigned by the SAT solver. Thus, one can determine the truth value of the spatial terms $\mu_{j,i}^*$ by calling an SAT solver for $\wedge_{j}^{p_1}(\vee_{i}^{p_2}\mu_{j,i}^*)$. The algorithm is summarized in Algorithm \ref{algorithm:verfier}.

\begin{algorithm}
\SetAlgoLined
\SetKwInOut{Input}{input}\SetKwInOut{Output}{output}
\Input{A set of GSTL formula $\Sigma=\{\varphi_1,\varphi_2,...\varphi_n\}$ and inference rules}
\Output{Consistency of $\Sigma$}
\BlankLine

 \While{there are $\Box_{[a,b]}$, $\Diamond_{[a,b]}$ and $\sqcup_{[a,b]}^*$ operators in formulas of $\Sigma$}{
 For GSTL formula $\varphi=\Box_{[a,b]}\phi$, we have $\varphi=\wedge_{i=a}^b\phi_i$ where $\phi_i$ represents $\phi$ at time $i\in[a,b]$\;
 For GSTL formula $\varphi=\Diamond_{[a,b]}\phi$, we have $\varphi=\vee_{i=a}^b\phi_i$ where $\phi_i$ represents $\phi$ at time $i\in[a,b]$\;
 For GSTL formula $\varphi_1\sqcup_{[a,b]}^o\varphi_2$, we have $\wedge_{i=a}^{b}(\varphi_{1,i}\wedge\varphi_{2,i})\wedge \varphi_{1,a-1} \wedge \varphi_{2,b+1} \wedge \neg\varphi_{1,b+1} \wedge \neg\varphi_{2,a-1}$ (other IA relations can be transferred in the similar way)\;
 }
 Reform $\Sigma$ into the CNF form $\wedge_{j=1}^{p_j}(\vee_{i=1}^{p_i}\mu_{j,i}^*)$ using the inference rules in \eqref{axiom schemas1}-\eqref{inference rules} (e.g., double negation elimination P1 and De Morgan's laws P9 and P10, and spatial axiom schema S1-S6) where $\mu_{j,i}^*$ are the spatial terms defined in \eqref{complexity proof 1}\;
 Solve the SAT problem for $\wedge_{j=1}^{p_j}(\vee_{i=1}^{p_i}\mu_{j,i}^*)$ by assigning a set of truth value $u:\tau\rightarrow\{\top,\bot\}\in\mathbf{U}$ to each $\mu_{j,i}^*$\;
 
 \eIf{a feasible $u$ has been found}{
   Output $\Sigma$ is consistent\;
   }{
   Output $\Sigma$ is inconsistent \;
  }
 \caption{Implementation of the automated reasoning}
 \label{algorithm:verfier}
\end{algorithm}

We consider the following example in Fig. \ref{example of moving a cup to a plate} to illustrate the implementation of the automated reasoning where a hand places a cup on top of a plate. In \eqref{example of moving a cup to a plate}, $\varphi_1$ describes a hand grabs a cup. $\varphi_2$ describes a cup that is on top of a plate. $\varphi_3$ describes an empty plate. $\varphi_4$ describes the action of a hand, grabbing a cup and putting it on an empty plate. $\varphi_5$ describes the fact that a cup on a plate will always lead to the plate being non-empty. $\varphi_6$ states the fact that a plate is empty for some time. $\Sigma=\{\varphi_4,\varphi_5,\varphi_6\}$ describes an action of moving a cup onto an empty cup with constraints that need to be considered (e.g., the plate must be empty before one can put a cup on it). Now we want to check if $\Sigma$ with temporal parameters given in Eq. \eqref{example1} is consistent.

\begin{figure}[H]
    \centering
    \includegraphics[scale=0.5]{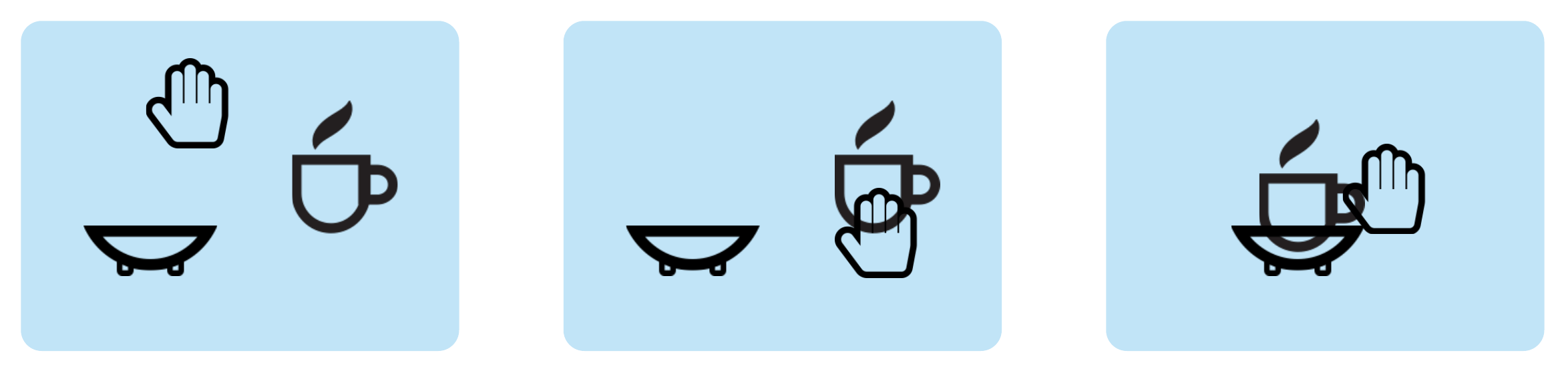}
    \caption{Example of moving a cup to a plate}
    \label{example of moving a cup to a plate}
\end{figure}

\begin{equation}
    \begin{aligned}
            &\varphi_1=\mathbf{C}_\exists(hand\wedge\mathbf{N}_\exists^{\left \langle *, *, * \right \rangle}cup)\\
            &\varphi_2=\mathbf{C}_\exists(cup\wedge\mathbf{N}_\exists^{\left \langle e, e, m \right \rangle}plate)\\
            &\varphi_3=\mathbf{C}_\exists(plate\wedge\mathbf{N}_\exists^{\left \langle e, e, m^{-1} \right \rangle}empty)\\
            &\varphi_4=\Box_{[10,20]}\varphi_1\sqcup_{[15,20]}^o\Box_{[15,25]}\varphi_2 \\
            &\varphi_5=\Box_{[0,+\infty)}\neg(\varphi_2\wedge\varphi_3)\\
            &\varphi_6=\Box_{[0,15]}\varphi_3
    \end{aligned}
    \label{example1}
\end{equation}
Following the procedure in Algorithm \ref{algorithm:verfier}, we first reformulate the GSTL formula in $\Sigma$ by eliminating all temporal operators. For example, $\varphi_4$ can be reformulated as
\begin{equation}
\begin{aligned}
\varphi_4=
    &\left(\bigwedge^{20}_{10}\varphi_1\right) 
    \left(\bigwedge^{20}_{15}(\varphi_1\wedge \varphi_2)\right)
    \left(\bigwedge^{25}_{15}\varphi_2\right).
\end{aligned}
\label{CNF1}
\end{equation}
As we can see from \eqref{CNF1}, temporal operators are replaced with conjunction operators with temporal parameters. To verify each clause in \eqref{CNF1}, we use the axiom schemas in \eqref{inference rules} (e.g. P9 and S1-S6 in this example) and obtain the following CNF form for $\varphi_1$.
\begin{equation}
\begin{aligned}
    &\varphi_1=\mathbf{C}_\exists(hand\wedge\mathbf{N}_\exists^{\left \langle *, *, * \right \rangle}cup)
    =\bigvee_{j=1}^{n_j}\left(\tau_j\wedge\upsilon_j\right)
    =\bigwedge\begin{pmatrix}
\tau_1\vee\tau_2\vee\cdots\vee\tau_{n_j-1}\vee\tau_{n_j} & \tau_1\vee\tau_2\vee\cdots\vee\tau_{n_j-1}\vee\upsilon_{n_j} &\\ 
\vdots&\vdots  &  \\ 
\tau_1\vee\upsilon_2\vee\cdots\vee\tau_{n_j-1}\vee\tau_{n_j} & \tau_1\vee\upsilon_2\vee\cdots\vee\tau_{n_j-1}\vee\upsilon_{n_j}&\\ 
\vdots&\vdots  &  \\  
\upsilon_1\vee\upsilon_2\vee\cdots\vee\upsilon_{n_j-1}\vee\tau_{n_j}  & \upsilon_1\vee\upsilon_2\vee\cdots\vee\upsilon_{n_j-1}\vee\upsilon_{n_j} &
\end{pmatrix}
\end{aligned}
\label{CNF2}
\end{equation}
where $\tau_j=\mathbf{C}_{A_j}hand$ and $\upsilon_j=\bigvee_{k=1}^{n_k}\mathbf{C}_{A_j}\mathbf{N}_{A_k}^{\left \langle *, *, * \right \rangle}cup$ whose truth value are to be assigned by the SAT solver. We can follow the same procedure to encode all other formula and obtain a set of formula in CNF $\Sigma=\{\wedge_{j=1}^{p_j}(\vee_{i=1}^{p_i}\mu_{j,i}^*)\}$. Then we use an SAT solver to assign a truth value to all spatial terms $\mu_{j,i}^*$.  We implement the simulation with an SAT solver in Python by treating each spatial term at a different time as a variable and choosing $n_j=n_k=1$ based on Fig. \ref{example of moving a cup to a plate}. We choose the maximum time length as 25. Following the SAT encoding procedure in Algorithm \ref{algorithm:verfier}, the resulting SAT problem has 150 variables. We have the following results from the SAT solver. For temporal parameters in \eqref{example1}, there is no feasible solution since the SAT solver cannot find a solution for spatial terms $\varphi_2$ and $\varphi_3$ at time 15. This can be easily verified in \eqref{example1} since $\varphi_2$ is true at time 15, according to $\varphi_4$ and $\varphi_3$ is true at time 15 according to $\varphi_6$; however, $\varphi_5$ requires they cannot both be true at the same time. If we change $\varphi_6$ to $\varphi_6=\Box_{[0,14]}\varphi_3$, then the SAT solver can return 16384 feasible solutions which means $\Sigma$ is consistent.

\section{Conclusion}\label{Section: conclusion}
Motivated by the importance of knowledge representation and automated reasoning and the limitation in existing spatial temporal logics, we propose a new graph-based spatial temporal logic by proposing a new ``until" temporal operator and three spatial operators with interval algebra extension. The satisfiability problem of the proposed GSTL is decidable. A Hilbert style inference system is given with a set of axiom schemas and two inference rules. We prove that the corresponding axiomatization system is sound and complete, and automated reasoning can be implemented through an SAT solver.


\bibliography{mybibfile}

\begin{thebibliography}{10}
\expandafter\ifx\csname url\endcsname\relax
  \def\url#1{\texttt{#1}}\fi
\expandafter\ifx\csname urlprefix\endcsname\relax\def\urlprefix{URL }\fi
\expandafter\ifx\csname href\endcsname\relax
  \def\href#1#2{#2} \def\path#1{#1}\fi

\bibitem{smoliar1994content}
S.~W. Smoliar, H.~Zhang, Content based video indexing and retrieval, IEEE
  multimedia 1~(2) (1994) 62--72.

\bibitem{hirschberg2015advances}
J.~Hirschberg, C.~D. Manning, Advances in natural language processing, Science
  349~(6245) (2015) 261--266.

\bibitem{fikes1971strips}
R.~E. Fikes, N.~J. Nilsson, Strips: A new approach to the application of
  theorem proving to problem solving, Artificial intelligence 2~(3-4) (1971)
  189--208.

\bibitem{hayes1983building}
F.~Hayes-Roth, D.~A. Waterman, D.~B. Lenat, Building expert system.

\bibitem{kifer1995logical}
M.~Kifer, G.~Lausen, J.~Wu, Logical foundations of object-oriented and
  frame-based languages, Journal of the ACM (JACM) 42~(4) (1995) 741--843.

\bibitem{googleknowledgegraph}
{Google}, Google knowledge graph,
  \url{https://developers.google.com/knowledge-graph} (2019).

\bibitem{liu2004conceptnet}
H.~Liu, P.~Singh, Conceptnet—a practical commonsense reasoning tool-kit, BT
  technology journal 22~(4) (2004) 211--226.

\bibitem{hertzberg2008ai}
J.~Hertzberg, R.~Chatila, Ai reasoning methods for robotics, Springer handbook
  of robotics (2008) 207--223.

\bibitem{mccarthy1960programs}
J.~McCarthy, Programs with common sense, RLE and MIT computation center, 1960.

\bibitem{post1921introduction}
E.~L. Post, Introduction to a general theory of elementary propositions,
  American journal of mathematics 43~(3) (1921) 163--185.

\bibitem{baader2003description}
F.~Baader, D.~Calvanese, D.~McGuinness, P.~Patel-Schneider, D.~Nardi, The
  description logic handbook: Theory, implementation and applications,
  Cambridge university press, 2003.

\bibitem{baier2008principles}
C.~Baier, J.-P. Katoen, Principles of model checking, MIT press, 2008.

\bibitem{cohn2001qualitative}
A.~G. Cohn, S.~M. Hazarika, Qualitative spatial representation and reasoning:
  An overview, Fundamenta informaticae 46~(1-2) (2001) 1--29.

\bibitem{raman2015reactive}
V.~Raman, A.~Donz{\'e}, D.~Sadigh, R.~M. Murray, S.~A. Seshia, Reactive
  synthesis from signal temporal logic specifications, in: Proceedings of the
  18th International Conference on Hybrid Systems: Computation and Control
  (HSCC), ACM, 2015, pp. 239--248.

\bibitem{kontchakov2007spatial}
R.~Kontchakov, A.~Kurucz, F.~Wolter, M.~Zakharyaschev, Spatial logic+ temporal
  logic=?, in: Handbook of spatial logics, Springer, 2007, pp. 497--564.

\bibitem{haghighi2016robotic}
I.~Haghighi, S.~Sadraddini, C.~Belta, Robotic swarm control from
  spatio-temporal specifications, in: 2016 IEEE 55th Conference on Decision and
  Control (CDC), IEEE, 2016, pp. 5708--5713.

\bibitem{bartocci2017monitoring}
E.~Bartocci, L.~Bortolussi, M.~Loreti, L.~Nenzi, Monitoring mobile and
  spatially distributed cyber-physical systems, in: Proceedings of the 15th
  ACM-IEEE International Conference on Formal Methods and Models for System
  Design, ACM, 2017, pp. 146--155.

\bibitem{cohn1997qualitative}
A.~G. Cohn, B.~Bennett, J.~Gooday, N.~M. Gotts, Qualitative spatial
  representation and reasoning with the region connection calculus,
  GeoInformatica 1~(3) (1997) 275--316.

\bibitem{aiello2007handbook}
M.~Aiello, I.~Pratt-Hartmann, J.~van Benthem, et~al., Handbook of spatial
  logics, Vol.~4, Springer, 2007.

\bibitem{allen1984towards}
J.~F. Allen, Towards a general theory of action and time, Artificial
  intelligence 23~(2) (1984) 123--154.

\bibitem{haghighi2015spatel}
I.~Haghighi, A.~Jones, Z.~Kong, E.~Bartocci, R.~Gros, C.~Belta, Spatel: a novel
  spatial-temporal logic and its applications to networked systems, in:
  Proceedings of the 18th International Conference on Hybrid Systems:
  Computation and Control, 2015, pp. 189--198.

\bibitem{smith1992algebraic}
T.~R. Smith, K.~K. Park, Algebraic approach to spatial reasoning, International
  Journal of Geographical Information Systems 6~(3) (1992) 177--192.

\bibitem{ligozat1996new}
G.~Ligozat, A new proof of tractability for ord-horn relations, in: AAAI/IAAI,
  Vol. 1, 1996, pp. 395--401.

\bibitem{liu2017distributed}
Z.~Liu, B.~Wu, J.~Dai, H.~Lin, Distributed communication-aware motion planning
  for multi-agent systems from stl and spatel specifications, in: 2017 IEEE
  56th Annual Conference on Decision and Control (CDC), IEEE, 2017, pp.
  4452--4457.

\bibitem{cook1971complexity}
S.~A. Cook, The complexity of theorem-proving procedures, in: Proceedings of
  the third annual ACM symposium on Theory of computing, 1971, pp. 151--158.

\bibitem{gabbay1990axiomatization}
D.~M. Gabbay, I.~M. Hodkinson, An axiomatization of the temporal logic with
  until and since over the real numbers, Journal of Logic and Computation 1~(2)
  (1990) 229--259.

\bibitem{cohn2008qualitative}
A.~G. Cohn, J.~Renz, Qualitative spatial representation and reasoning,
  Foundations of Artificial Intelligence 3 (2008) 551--596.

\end{thebibliography}

\end{document}